\documentclass[11pt]{article}

\usepackage{amssymb,amsmath,amsthm,mathrsfs}
\newcommand{\R}{\mathbb{R}}

\newcommand{\pd}{\partial}
\newcommand{\N}{\mathbb{N}}

\newcommand{\Sp}{\mathbb{S}}

\newcommand{\Hilbert}{\mathcal{H}}

\newcommand{\supp}{\mathrm{supp}}

\newtheorem{Lemma}{Lemma}
\newtheorem{Theorem}{Theorem}

\newtheorem{Corollary}{Corollary}

\theoremstyle{definition}
\newtheorem{Remark}{Remark}
\newtheorem{ass}{Assumption}

\begin{document}

\title{\bf Estimates on trapped modes in deformed quantum layers}

\author{
Hynek Kova\v{r}{\'\i}k and Semjon Vugalter
}
\date{
\begin{center} {\small 
Institute of Analysis, Dynamics
and Modeling, Universit\"at Stuttgart, 
PF 80 11 40, D-70569  Stuttgart, Germany.
}
\end{center}
}

%

\maketitle

\begin{abstract}
\noindent We use the logarithmic Lieb-Thirring inequality for
two-dimensional Schr\"odinger operators and establish estimates on
trapped modes in geometrically deformed quantum layers.
\end{abstract}


\section{Introduction}
Trapped modes in quantum layers and waveguides have been intensively
studied in the last decades, see \cite{BEGK,BGRS,DE,DEK,ES,GJ} and
references therein. In these papers it has been shown that a
suitable geometrical perturbation of a waveguide (or a layer)
$\Omega$, such as local enlargement or bending, induces the existence
of discrete eigenvalues $E_j$ of the corresponding Laplace operator
\begin{equation*}
-\Delta_\Omega \quad \text{in} \quad L^2(\Omega)
\end{equation*}
with Dirichlet boundary conditions. These eigenvalues represent the
so called {\it trapped modes}, 
which are the main objects of our interest.
For mildly deformed waveguides and layers the corresponding weak coupling
behaviour of such eigenvalues have been established in
\cite{BEGK,BGRS,DE,DEK}.

The next step in the analysis of the above mentioned eigenvalues
consist of deriving suitable spectral estimates. In other words,
one would like to know not only that these eigenvalues exist, but also 
in which way they are linked to the deformation of $\Omega$, i.e.~how 
the distance of $E_j$ to the essential
spectrum of $-\Delta_\Omega$ depends on the perturbation. 
Such a connection can be formulated
in terms of certain Lieb-Thirring type inequalities,
which estimate the sums 
\begin{equation} \label{riesz}
\sum_j\, |E-E_j|^\gamma,\qquad E:= \inf\sigma_{ess}
(-\Delta_\Omega)\, , \quad \gamma \geq 0\, .
\end{equation}
In the case in which $\Omega$ is a quantum waveguide, these estimates
were proved in \cite{EW} for potential type perturbations and in
\cite{ELW} for geometrical perturbations and perturbations of the
boundary conditions. In the case of a quantum
layer with a potential perturbation, the corresponding inequality was
recently obtained in \cite{KVW}. All these estimates have the right
order of asymptotics for weak perturbations, i.e.~the respective
upper bounds on the sum \eqref{riesz} reflect the correct weak
coupling behaviour established in \cite{BEGK,BGRS,DE,Si}.

\vspace{0.15cm}

The aim of the present paper is to extend these results also to the
case of a geometrical deformation of a quantum layer. We note that
in the case of quantum waveguides the key ingredient of the proof of
an estimate, which has the correct asymptotical behaviour, was the
Lieb-Thirring inequality for one-dimensional Schr\"odinger operators
with the critical power $\gamma=\frac 12$ proved in \cite{W}. Since
a layer might be considered as a two-dimensional analog of a
waveguide, the key ingredient of our proof will be the corresponding
logarithmic critical Lieb-Thirring inequality for two-dimensional Schr\"odinger
operators, which was recently established in \cite{KVW}. Therefore we first
briefly recall the result of \cite{KVW}, see Theorem \ref{2dim}.
In section \ref{bulges} we then show how the problem can be
reduced to the spectral analysis of certain two-dimensional
Schr\"odinger operator with the effective potential induced by the
geometrical deformation of the layer.

\vspace{0.15cm}

\noindent Following notation will be adopted in the text. 
Given a Hilbert space $\Hilbert$ and a self-adjoint operator
$T$ in $\Hilbert$ we denote by $N_\Hilbert(T)$ the number of negative 
eigenvalues of $T$, counting their geometrical multiplicities. When necessary
we will use the symbols $\Delta_{x,y}, \, \nabla_{x,y}$ etc.~in order to
specify in which variables the respective operators act.


\section{Preliminaries}


\subsection{Quantum layers}

A quantum layer may be represented by an open domain 
$\Omega=\R^2\times (0,d)$, more precisely $\Omega:= \{x,y,z\in
\R^3\, :\, 0< z < d \}$, where $d$ is the width of $\Omega$. It will
be convenient to work with the shifted Laplace operator
\begin{equation}
A = -\Delta_{\Omega} -\frac{\pi^2}{d^2}\quad \text{in} \quad
L^2(\Omega)
\end{equation}
with the Dirichlet boundary conditions at $\pd\Omega$. The operator
$A$ is associated with closed quadratic form
\begin{equation} \label{form}
Q[u] = \int_{\Omega}\, \left(|\nabla u|^2 -\frac{\pi^2}{d^2}\,
  |u|^2\right)\, dx dy dz
\end{equation}
with the form domain $H^1_0(\Omega)$. It can be easily verified that
$$
\sigma_{ess}\, (A) =[0,\infty), \quad \sigma_{d}\, (A) = \emptyset\,
.
$$
As noted in \cite{BEGK}, a local enlargement of the width of the
layer will not affect the essential spectrum of $A$, but will lead
to the existence of negative discrete eigenvalues of $A$.
To find a suitable spectral estimate on these eigenvalues we need  the
two-dimensional logarithmic Lieb-Thirring inequality, which we
formulate in the next section.


\subsection{Two-dimensional Lieb-Thirring inequality}
Consider the Schr\"odinger operator
\begin{equation}
-\Delta - V \quad \text{in}\quad L^2(\R^2)\, ,
\end{equation}
where $V$ is a potential function decaying at infinity and such that
$\sigma_{ess}(-\Delta - V)=[0,\infty)$. Denote by $-\lambda_j$ the
negative eigenvalues of $-\Delta - V$ and introduce
the family of functions $F_s:(0,\infty)\to(0,1]$ defined by
\begin{equation}
\forall\, s>0 \qquad F_s(t) := \left\{
\begin{array}{l@{\quad \mathrm{} \quad }l}
 |\ln ts^2|^{-1} & 0< t \leq e^{-1}s^{-2}\, , \\
 &  \\
1 & t > e^{-1}s^{-2} \, .
\end{array}
\right.
\end{equation}
An upper bound on the sum
$$
\sum_j\, F_s(\lambda_j)
$$
in terms of intergals of $V$ 
has been recently found in \cite{KVW}. Its formulation requires
some additional notation. The space $L^1(\R_+,L^p(\Sp^1))$
is defined as the space of functions $f$ such that
\begin{equation} \label{L1p}
\|f\|_{L^1(\R_+, L^p(\Sp^1))} : = \int_0^{\infty}
\left(\int_0^{2\pi}
  |f(r,\theta)|^p\, d\theta\right)^{1/p}\, r\, dr \, < \infty \, ,
\end{equation}
where $(r,\theta)$ are the polar coordinates in $\R^2$.
Moreover, given an $s>0$ we introduce $B(s):=\{x\in \R^2\, :\, |x| < s \}$.
The result of \cite{KVW} then reads as follows:

\begin{Theorem} \label{2dim}
Let $V\geq 0$ and $V\in L^1_{loc}(\R^2,|\ln|x||\, dx)$. Assume that
$V\in L^1(\R_+,L^p(\Sp^1))$ for some $p>1$.
Then the eigenvalues $-\lambda_j$ satisfy the inequality
\begin{equation} \label{first-ineq}
\sum_j\, F_s(\lambda_j) \, \leq \, c_1\, \|V\, \ln
(|x|/s)\|_{L^1(B(s))}\, + c_p\, \|V\|_{L^1(\R_+, L^p(\Sp^1))}\,
\end{equation}
for all $s\in\R_+$. The constants $c_1$ and $c_p$ are independent of
$s$ and $V$.

\vspace{0.1cm}

\noindent
In particular, if $V(x) = V(|x|)$, then there exists a constant $c_4$, such that
\begin{equation} \label{radial}
\sum_j\, F_s(\lambda_j)  \, \leq \,  c_1\,
\|V\ln(|x|/s)\|_{L^1(B(s))}\, + c_4\, \|V\|_{L^1(\R^2)}
\end{equation}
holds true for all $s\in\R_+$.
\end{Theorem}

\vspace{0.15cm}

\noindent Note that for weak potentials $V$ the estimate
\eqref{first-ineq} reflects the exponential asymptotical
behaviour of the lowest eigenvalue of $-\Delta-V$ established in \cite{Si}.
Since the behaviour of weakly coupled eigenvalues in a layer is 
essentially two-dimensional, the corresponding asymptotics for weakly
deformed layers is again of the exponential type, see \cite{BEGK}.
Our goal thus is to find a similar upper bound for geometrical induced
eigenvalues in quantum layers.


\section{A layer with a geometrical perturbation}
\label{bulges}

Here we apply Theorem \ref{2dim} to obtain the estimates
on the discrete eigenvalues of the Dirichlet Laplacian in a
layer whose width is locally enlarged; 
$$
\Omega_f:=
\{x,y,z\in \R^3\, :\, 0< z < d+ f(x,y)\} ,
$$ 
where $f:\R^2\to [0,\infty)$. We consider the shifted Laplace operator
\begin{equation}
A_f = -\Delta_{\Omega_f} -\frac{\pi^2}{d^2}\quad \text{in} \quad L^2(\Omega_f)
\end{equation}
with the Dirichlet boundary conditions at $\pd\Omega_f$ which is 
associated with the closed quadratic form
\begin{equation} \label{form_f}
Q_f[u] = \int_{\Omega_f}\, \left(|\nabla u|^2 -\frac{\pi^2}{d^2}\,
  |u|^2\right)\, dx
\end{equation}
with the form domain $H^1_0(\Omega_f)$.
From the assumptions on $f$ follows that
$$
\sigma_{ess}\, (A_f) =[0,\infty)\, .
$$
Let us denote by $-\mu_j$ the
non decreasing sequence of negative eigenvalues of $A_f$ taking into
account their multiplicities. We shall estimate the total number of 
$-\mu_j$ by the number of negative eigenvalues of a certain
two-dimensional Schr\"odinger operator $-\Delta-V_f$
with $V_f$ depending on the deformation function $f$.

\begin{Theorem} \label{compare}
Assume that the function $f:\R^2\to\R$ is in $C^2(\R^2)$ 
and such that $\supp f\subset B(R)$ for some $R>0$, and
$\|f\|_\infty < d$.
For any $t\geq 0$ we have
\begin{equation} \label{number}
N_{L^2(\Omega_f)}(A_f-t) \,  \leq \, N_{L^2(\R^2)}(-\Delta+3V_f-3t)\, ,
\end{equation}
where
$$
V_f = \frac{\pi^2}{(d+f)^2}\, -\frac{\pi^2}{d^2}
-b_1|\nabla f|^2\, -b_2(R)\, |\Delta f|^2\, -b_3(R)\, |\nabla f|^4 \, ,
$$
with $b_1, b_2(R)$ and $b_3(R)$ satisfying \eqref{upper}.
\end{Theorem}
\begin{proof}
We write a given trial function $\psi\in H_0^1(\Omega_f)$ as
\begin{equation} \label{decomp}
\psi(x,y,z) = \varphi(x,y,z)\, g(x,y) + h(x,y,z) \, ,
\end{equation}
where
$$
\varphi(x,y,z) = \sqrt{\frac{2}{d+f(x,y)}}\, \, \sin\left(\frac{\pi
  z}{d+f(x,y)}\right) \,
$$
and
\begin{equation} \label{ort}
\int_0^{d+f(x,y)}\, \varphi(x,y,z)h(x,y,z)\, dz = 0\quad \forall\,
(x,y)\in\R^2\, .
\end{equation}
Hence
\begin{align} \label{qform}
& \int_{\Omega_f}\, \left(|\nabla\psi|^2-\frac{\pi^2}{d^2}\,
  |\psi|^2\right)\, dx\, dy\, dz  = \int_{\Omega_f}\, \big(|\nabla\varphi|^2|g|^2
+|\nabla_{x,y}\, g|^2+|\nabla h|^2   \nonumber \\
& \qquad \qquad-\frac{\pi^2}{d^2}\, (|\varphi\, g|^2+|h|^2)+ 2g g'_x
\varphi'_x\varphi +2g\varphi'_x h'_x+2\varphi g'_x h'_x + 2g
g'_y\varphi'_y\varphi \nonumber \\
& \qquad  \qquad +2g\varphi'_y h'_y+2\varphi g'_y h'_y
+2g\varphi'_z h'_z  \big)\, dx\, dy\, dz\, .
\end{align}
Here and in the sequel we will use the shorthands $u'_x = \frac{\pd
u}{\pd x}$ and analogously for other partial derivatives. 
We estimate all the mixed terms in \eqref{qform},
except for the last two, point-wise in the following way:
\begin{eqnarray} \label{pointwise}
2g\, g'_x\, \varphi'_x\, \varphi & \leq & a_1^{-1}\, |\varphi\, g'_x|^2+a_1\,
|g\varphi'_x|^2\, , \nonumber \\
2g\, g'_y\, \varphi'_y\, \varphi & \leq &a_1^{-1}\, |\varphi\, g'_y|^2+a_1\,
|g\varphi'_y|^2\, , \nonumber \\
2g \varphi'_x\, h'_x & \leq & a_2^{-1}\, |h'_x|^2+a_2\,
|g\varphi'_x|^2\, , \nonumber \\
2g\, \varphi'_y\, h'_y & \leq & a_2^{-1}\, |h'_y|^2+a_2\,
|g\varphi'_y|^2\, ,
\end{eqnarray}
where $a_1$ and $a_2$ are real positive numbers whose values will be
specified later. Furthermore, from integration by parts and
\eqref{ort} follows that
$$
\int_{\Omega_f}\, g\varphi'_z h'_z\, dxdydz = -\int_{\Omega_f}\, g\varphi''_z h
\, dxdydz = 0\, .
$$
integrating by parts again and using \eqref{ort} we can rewrite the last two
terms in \eqref{qform} as
\begin{align*}
\int_{\Omega_f} \, \varphi\, h'_x\, g'_x\, dxdydz & = -\int_{\Omega_f} \,
\varphi'_x\, h\, g'_x \, dxdydz = \int_{\Omega_f} \, g(\varphi''_x\,
h+\varphi'_x\, h'_x)\, dxdydz \, , \\
\int_{\Omega_f} \, \varphi\, h'_y\, g'_y\, dxdydz & = -\int_{\Omega_f} \,
\varphi'_y\, h\, g'_y \, dxdydz = \int_{\Omega_f} \, g(\varphi''_y\,
h+\varphi'_y\, h'_y)\, dxdydz\, .
\end{align*}
The terms $2g \varphi'_x\, h'_x$ and $2g \varphi'_y\, h'_y$ will be
estimated in the same way as in \eqref{pointwise}.
For the rest we use the following point-wise inequalities
\begin{eqnarray}
2g\, \varphi''_x\, h & \leq & a_3\, g^2 |\varphi''_x|^2+a_3^{-1}\,
h^2\, \chi_f\, , \nonumber \\
2g\, \varphi''_y\, h & \leq & a_3\, g^2 |\varphi''_y|^2+a_3^{-1}\,
h^2\, \chi_f\, , \nonumber
\end{eqnarray}
where $\chi_f$ denotes the characteristic function of the support of $f$.
Now we put $a_1=a_2=3$ and arrive at
\begin{eqnarray} \label{mainestim}
&& \int_{\Omega_f}\, \left(|\nabla\psi|^2-\frac{\pi^2}{d^2}\,
  |\psi|^2\right)\, dx\, dy\, dz  \geq \int_{\R^2}\, \left(
\frac 13\, |\nabla_{x,y}\, g|^2 + \tilde{V}_f(x,y) |g|^2\right) \, dx\, dy\, ,
 \nonumber \\
&& \quad + \int_{\Omega_f}\, \left(\frac13\, |\nabla_{x,y}\, h|^2 +|h'_z|^2
 -\frac{\pi^2}{d^2}\, |h|^2-
a_3^{-1}\, |h|^2\, \chi_f\right)\, dx\, dy\, dz
\end{eqnarray}
with
\begin{align*}
\tilde{V}_f & = \frac{\pi^2}{(d+f)^2}\, -\frac{\pi^2}{d^2}
- \int_0^{d+f}\, \left(5\left(|\varphi'_x|^2+|\varphi'_y|^2\right)
+a_3\left(|\varphi''_x|^2+|\varphi''_y|^2\right)\right)\, dz \, .
\end{align*}
Since $h$ satisfies Dirichlet boundary conditions at $\pd\Omega_f$ and
$f< d$, we deduce from \eqref{ort} that
\begin{eqnarray} \label{cylindr}
&& \int_{\Omega_f}\, \left(\frac13\, |\nabla_{x,y}\, h|^2 +|h'_z|^2
 -\frac{\pi^2}{d^2}\, |h|^2-
a_3^{-1}\, |h|^2\, \chi_f\right)\, dx\, dy\, dz \nonumber \\
&& \geq \int_{\Omega_f}\, \left(\frac 13\, |\nabla_{x,y}\, h|^2
  +\left(\frac{4\pi^2}{(d+f)^2}
-\frac{\pi^2}{d^2}\right)\, |h|^2-
a_3^{-1}\, |h|^2\, \chi_f \right)\, dx\, dy\, dz \nonumber \\
&&
\geq \int_0^d\int_{\R^2}\,\left(\frac 13\, |\nabla_{x,y}\, h|^2
+ \frac{3\pi^2}{d^2} \, |h|^2\, -
\left(a_3^{-1}+\frac{3\pi^2}{d^2}\right)
\, |h|^2\, \chi_f \right)\, dx\, dy\, dz  \nonumber \\
&& + \int_d^{2d}\int_{\supp f}\,\left(\frac 13\, |\nabla_{x,y}\, h|^2
 - a_3^{-1}\, |h|^2\,\right)\, dx\, dy\, dz
\end{eqnarray}
From the fact that the support of $f$ is compact it follows that the
last term in 
\eqref{cylindr}
is non-negative for all $a_3 \geq \lambda^{-1}(R)$, where $\lambda(R)$
is the lowest eigenvalue of $-\Delta_{x,y}$ on the disc $B(R)$ with
Dirichlet boundary conditions. Moreover,
the expression on the third line of \eqref{cylindr} can be bounded
from below as follows
\begin{eqnarray} \label{disc}
&&
\int_0^d\int_{\R^2}\,\left(\frac 13\, |\nabla_{x,y}\, h|^2
+ \frac{3\pi^2}{d^2} \, |h|^2\, -
\left(a_3^{-1}+\frac{3\pi^2}{d^2}\right)
\, |h|^2\, \chi_f \right)\, dx\, dy\, dz \\
&& \geq
\int_0^{d} \left(\int_{\R^2}\,
\left(\frac 13\, |\nabla_{r,\theta}\, h|^2
+\frac{3\pi^2}{d^2} \, |h|^2\, \chi_{[R,\infty)}
 - a_3^{-1}\, |h|^2\, \chi_{[0,R]}\right)\, r dr\, d\theta\right)\,
dz\, , \nonumber
\end{eqnarray}
where we have used the polar coordinates $(r,\theta)$ in $\R^2$.
In view of Lemma \ref{lem1}, see Appendix,
\eqref{disc} is positive for $a_3\geq \frac{d^2}{3\pi^2}$. Therefore we
choose
$$
a_3(R) = \max\left\{\frac{d^2}{3\pi^2},\, \lambda^{-1}(R)\right\}\, .
$$
Now it remains to estimate the first term on the right hand side of
\eqref{mainestim}. By a direct calculation we arrive at
\begin{align*}
& \int_0^{d+f}\, \left(5\left(|\varphi'_x|^2+|\varphi'_y|^2\right)
+a_3(R)\, \left(|\varphi''_x|^2+|\varphi''_y|^2\right)\right)\, dz \leq
b_1\, |\nabla f|^2 \\
&  \qquad \qquad \qquad \qquad  \qquad \qquad
+b_2(R)\, |\Delta f|^2+ b_3(R)\, |\nabla f|^4\,
\end {align*}
where $b_1,b_2(R),b_3(R)$ are positive numbers which satisfy
\begin{equation}\label{upper}
b_1\leq \frac{\pi^2}{2\, d^2}\, , \quad 
b_2(R) \leq \frac{a_3(R)\, \pi^2}{d^2}\, , \quad
b_3(R) \leq 4\, a_3(R)\, \left(\frac{\pi^2}{d^4}+\,
  \frac{\pi^4}{5d^2}\right)\, . 
\end{equation}
Finally, combining \eqref{mainestim} and \eqref{ort} we obtain
\begin{eqnarray*}
&& \int_{\Omega_f}\, \left(|\nabla\psi|^2-\frac{\pi^2}{d^2}\,
  |\psi|^2 - t |\psi|^2 \right)\, dx\, dy\, dz  \\
 && \quad \geq \frac 13\, \int_{\R^2}\, \left(
|\nabla_{x,y}\, g|^2 + 3 V_f(x,y) |g|^2 -3t |g|^2\right) \, dx\, dy\, ,
\end{eqnarray*}
holds true for any $t\leq 0$. In view of the variational principle and
the identity
$$
N_{L^2(\R^2)}\left(\frac 13\left(-\Delta+3V_f-3\tau\right)\right)\, = 
N_{L^2(\R^2)}(-\Delta+3V_f-3\tau)
$$
we conclude the proof.
\end{proof}

\begin{Remark} From the assumption $f<d$ it follows that all negative
eigenvalues of $A_f$ come from the first channel only. However, we
would like to mention that this assumption is purely 
is purely technical and could be
replaced by $f< n d, n\in \N$. In that case we would have to use another
decomposition of a test function $\psi$, analogous to \eqref{decomp},
taking into account also the functions associated with higher
transversal modes in $z$. 
For the sake of simplicity we therefore suppose $f<d$.
\end{Remark}

\begin{Corollary} \label{geom}
For any $p>1$ there exist positive constants $C_1$ and $C_p$
such that
\begin{equation} \label{mainineq1}
\sum_j F_s(\mu_j)  \,  \leq\, \, \,  C_1\,
\left\|V_f \ln (\sqrt{x^2+y^2}/s)\right\|_{L^1(B(s))} \,
+ C_p\, \|V_f\|_{L^1(\R_+, L^p(\Sp^1))}\,
\end{equation}
holds for all $s>0$.
\end{Corollary}

\begin{proof}
Since $F'_s$ is non-negative we have
\begin{eqnarray*}
\sum_j F_s(\mu_j) & = &\int_0^\infty\, F_s'(t)\,
N_{L^2(\Omega_f)}(A_f-t)\, dt  \\
& \leq &
\int_0^\infty\, F_s'(t)\, N_{L^2(\R^2)}(-\Delta+3V_f-3t)\, dt \\
& \leq &
3\int_0^\infty\, F_s'(t)\, N_{L^2(\R^2)}(-\Delta+3V_f-t)\, dt
= 3 \sum_j F_s(\lambda_j)\, .
\end{eqnarray*}
and the statement follows from Theorem \ref{2dim}.
\end{proof}

\noindent 
The disadvantage of estimate \eqref{mainineq1} is the presence of the
terms in $V_f$ which contain the derivatives of
$f$. Firstly, small oscillations of $f$ will lead to the unnecessary
growth of the right hand side in \eqref{mainineq1}. Secondly, the
deformation function $f$ in general need not be $C^2-$smooth. This can
remedied using the monotonicity property of eigenvalues of Laplace
operators in domains with Dirichlet boundary conditions. Namely, 
for any $\tilde{f}\geq f$ we have 
$$
N_{L^2(\Omega_f)}(A_f-t)  \, \, \leq \, \, 
N_{L^2(\Omega_{\tilde{f}})}(A_{\tilde{f}}-t) \qquad \forall\, t \geq 0\, .
$$

\noindent As an immediate consequence of Theorem \ref{2dim} and 
Corollary \ref{geom} we thus get 

\begin{Theorem} \label{radial}
Let $0\leq f < d$ be a continuous function with support in $B(R)$.
Then there exist constants $C_3$ and $C_4$ such that
\begin{equation} \label{geom-radial}
\sum_j F_s(\mu_j)  \,  \leq\, \, \,  \inf_{\tilde{f}\geq f}\, \left(C_3\,
\left\|V_{\tilde{f}} \ln\frac rs\right\|_{L^1(B(s))} \,
+ C_4\, \|V_ {\tilde{f}}\|_{L^1(\R^2)} \right)\, ,
\end{equation}
where the infimum is taken over all radially symmetric functions
$\tilde{f}\in C^2_0(B(R))$.
\end{Theorem}

\begin{Remark}
Let us consider the behaviour of the estimate \eqref{mainineq1} for
weakly deformed layers. This means replacing $f$ by $\alpha\,f$ and letting
$\alpha$ go to zero. Theorem \ref{compare} and the result of \cite{KVW} yield
the following upper bound on the number of negative eigenvalues of
$A_{\alpha f}$: 
\begin{equation*}
N_{L^2(\Omega_f)} (A_{\alpha f}) \, \, \leq\, \, 1 + \text{const}\,
\left (\left\|V_{\alpha\tilde{f}}\, \ln\frac rs\right\|_{L^1(\R^2)}
  \, + \|V_{\alpha \tilde{f}}\|_{L^1(\R^2)} \right)\, .  
\end{equation*}
From the explicit form of $V_{\alpha \tilde{f}}$ thus follows that $A_{\alpha
  f}$ has only one negative eigenvalue, $-\mu_1(\alpha)$, for $\alpha$ small
enough. Moreover, inequality \eqref{mainineq1} implies
\begin{equation} 
|\mu_1(\alpha)|\, \, \leq \, \,
\exp\left(-\frac{C(f,d)}{w(\alpha)}\right)\, ,
\end{equation}
where $C(f,d)$ is a positive factor independent of $\alpha$ and 
\begin{equation}
w(\alpha) = \alpha + \mathcal{O}(\alpha^2)\qquad \alpha\to 0\, . 
\end{equation}
This agrees, in order of $\alpha$, with the asymptotics found in \cite{BEGK}.
\end{Remark}

\section*{Appendix}

\begin{Lemma} \label{lem1}
Let $u\in H^1(\R_+, r\, dr)$. Then for any $a>0$ and any $R>0$ the
inequality
\begin{equation} \label{hardy}
a \int_0^R\, |u|^2\, r\, dr \leq a  \int_R^{2R}\, |u|^2\, r\,
dr+ \int_0^{2R}\, |u'|^2\, r\, dr
\end{equation}
holds true.
\end{Lemma}

\begin{proof}
Let us define the function $h:\R_+\to \R$ by
$$
h(r)=  \left\{\begin{array}{lll} \alpha & 0< r\leq R\\
 \alpha\left(1-\frac{r-R}{R}\right) & R<r<2R \\
0 & 2R \leq r
\end{array}\right. \, ,
$$
where $\alpha$ is a positive constant. For any $r\in (0,R)$ we then
have
\begin{eqnarray}
\alpha u(r) & =& h(r) u(r) = -\int_r^{2R}\, (hu)'(t)\, dt \\
& =& -\frac{\alpha}{R}\, \int_R^{2R}\, u\, dt - \int_r^{2R}\,
hu'\, dt\nonumber \, .
\end{eqnarray}
The Cauchy-Schwarz inequality thus implies
$$
\alpha^2\, |u(r)|^2 \, \leq \, 2\, \frac{\alpha^2}{R}\, \int_R^{2R}\, |u|^2\, dt
+ 2\|h\|^2\, \int_r^{2R}\,
|u'|^2\, dt\, .
$$
Multiplying by $r$ and integrating over $(0,R)$ we get
$$
\alpha^2\, \int_0^R\, |u|^2\, r\, dr \, \leq \, \alpha^2\,
\int_R^{2R}\, |u|^2\, r\,
dr + 2R\, \|h\|^2\, \int_0^{2R}\, |u'|^2\, r\, dr\, .
$$
To conclude the proof it suffices to choose $\alpha^2= 2R\, \|h\|^2\,
a$.
\end{proof}

\section*{Acknowledgement}
The support from the DFG grant WE 1964/2 is gratefully acknowledged.

%
%
\providecommand{\bysame}{\leavevmode\hbox to3em{\hrulefill}\thinspace}
\providecommand{\MR}{\relax\ifhmode\unskip\space\fi MR }
\providecommand{\MRhref}[2]{%
  \href{http://www.ams.org/mathscinet-getitem?mr=#1}{#2}
}
\providecommand{\href}[2]{#2}
\end{document}